\documentclass{amsart}

\usepackage{eqnarray,amsmath,mathtools,tensor, xypic, dsfont}

\usepackage{slashed}

\usepackage{csquotes}

\usepackage{cite}

\usepackage{tikz-cd}

\usepackage{url} 
\usepackage{hyperref}

\newtheorem{theorem}{Theorem}[section]

\theoremstyle{definition}
\newtheorem{definition}[theorem]{Definition}
\newtheorem{example}[theorem]{Example}

\theoremstyle{remark}
\newtheorem{remark}[theorem]{Remark}

\numberwithin{equation}{section}

\newcommand{\C}{{\mathbb{C}}}
\newcommand{\R}{{\mathbb{R}}}
\newcommand{\T}{{\mathbb{T}}}

\newcommand{\FF}{{\mathcal{F}}}
\newcommand{\AAA}{{\mathcal{A}}}

\newcommand{\Spaces}{{\mathbf{Spaces}}}
\newcommand{\NCSpaces}{{\mathbf{NCSpaces}}}
\newcommand{\CommAlgebras}{{\mathbf{Algebras}}}
\newcommand{\NCAlgebras}{{\mathbf{NCAlgebras}}}
\newcommand{\TTT}{{\mathcal{T}}}
\newcommand{\Stacks}{{\mathbf{Stacks}}}
\newcommand{\Groupoids}{{\mathbf{Groupoids}}}

\newcommand{\GL}{{\mathrm{GL}}}
\newcommand{\Li}{{\mathrm{Li}}}

\newcommand{\per}{\mathrm{per}}
\newcommand{\even}{\mathrm{even}}
\newcommand{\odd}{\mathrm{odd}}

\newcommand{\topo}{\mathrm{top}}
\newcommand{\Img}{\mathrm{Im\ }}
\newcommand{\Ker}{\mathrm{Ker\ }}
\newcommand{\op}{\mathrm{op}}

\newcommand{\redu}{\mathrm{red}}

\newcommand{\OO}{\ensuremath{{\mathcal O}}}
\newcommand{\GG}{\ensuremath{{\mathcal G}}}
\newcommand{\HH}{\ensuremath{{\mathcal H}}}
\newcommand{\UU}{\ensuremath{{\mathcal U}}}
\newcommand{\VV}{\ensuremath{{\mathcal V}}}

\newcommand{\DD}{\ensuremath{{\slashed \partial}}}

\newcommand{\integer}{\ensuremath{{\mathbb Z}}}
\newcommand{\naturals}{\ensuremath{{\mathbb N}}}
\newcommand{\real}{\ensuremath{{\mathbb R}}}
\newcommand{\complex}{\ensuremath{{\mathbb C}}}

\newcommand{\rational}{\ensuremath{{\mathbb Q}}}
\newcommand{\trop}{\ensuremath{{\mathbb T}}}

\newcommand{\adeles}{\ensuremath{{\mathbb A}}}
\newcommand{\field}{\ensuremath{{\mathbb K}}}

\begin{document}

\title{Non-commutative Geometry Indomitable}
\author{Ernesto Lupercio}
\address{Department of Mathematics, Cinvestav, Av. IPN 2508, Mexico
City, Mexico 07300}
\email{lupercio@math.cinvestav.mx}
\thanks{The author is professor of mathematics at the Center for Research and Advanced Studies of the National Polytechnic  Institute (Cinvestav, IPN) in Mexico City.}

\thanks{This paper is to appear in the \emph{Notices of the American Mathematical Society} later this year}

\date{January 30, 2020 and, in revised form, August 15, 2020.}

\dedicatory{To Thomas Nevins, an indomitable geometer, in Memoriam.}

\keywords{Differential geometry, algebraic geometry, Toric varieties, Riemann Hypothesis}

\maketitle

\begin{abstract}
	This paper is a very brief and gentle introduction to non-commutative geometry geared primarily towards physicists and geometers. It starts with a brief historical description of the motivation for non-commutative geometry and then goes on to motivate the subject from the point of view of the the understanding of local symmetries affordee by the theory of groupoids. The paper ends with a very rapid survey of recent developments and applications such as non-commutative toric geometry, the standard model for particle physics and the study of the Riemann Hypothesis.   
\end{abstract}

\section{Introduction}

In was a monumental development in mathematics the discovery by Descartes that the (commutative) algebra, developed by the Persian civilization (al-Khwarismi, Omar Khayyam), and the towering edifice of Euclidean geometry -- the jewel of Greek mathematics -- were but two sides of the same coin. 

Thus, the publication in 1637 of \emph{La G\'eom\'etrie} (just an appendix to the \emph{Discourse on the Method}!) is one of the great masterpieces of mathematics literature; in it, a grand unification is put forward: that of (commutative) algebra with geometry. This development by Descartes did not come out of the blue, and one can trace most of the ideas of this treatise to various predecesors; nevertheless, it was in \emph{La G\'eom\'etrie} that the explicitness and clarity of the unification idea makes it take a life of its own.

Both, generalization and abstraction (more often than not, inspired and informed by physics) have always been a great source of inspiration for mathematicians, aiding us to overcome our prejudices and lack of imagination. And so, while today it sounds totally obvious the passage form the linearly ordered topologically complete field of the real numbers $\mathbf{R}$, to the algebraically closed field of the complex numbers $\mathbf{C}$ (where we have lost the ordering), and hence, just as obvious the passage from real geometry to complex geometry, it is not until the 1857 masterpiece by B. Riemann, \emph{Theorie der Abelschen Functionen}, that this new geometry enters into the conciousness of mathematicians at large. By losing the linear order of the field, we gain a theory of enormous beauty and coherence.

A different source for the development of mathematics has been the desire to solve ever more classes of equations; in fact, one could say that the evolution of the concept of number is inseparable from this development. 

From this perspective, just as complex numbers (which need two real components) appear motivated by the solution of the general quadratic equation, the introduction of (square) matrices (already appearing in the classical Chinese text \emph{The Nine-Chapters on the Mathematical Art}, 10th-2nd centuries BC), is motivated by the solution to the general system of linear equations (the word matrix used in this context first appears in 1850, not long before the paper of Riemann, in the work of J.J. Silvester). These remarks would suggest that one could think of the ring of matrices as a new kind of set of numbers, but numbers whose multiplication \emph{does not commute} (Gelfand). Again, by losing something (commutativity), in exchange, we gain something else (linear algebra first, then matrix groups, etc.). In fact, it is in A. Cayley's 1858 masterpiece. \emph{A Memoir on the Theory of Matrices}, where all this has come to fruition; there, Cayley proves the \emph{non-commutativity} property of matrix multiplication.

Just as classical mechanics is the inspiration behind Descartes, and electrostatics, the motivation behind Riemann, it is quantum physics, and its matrix non-commutative mathematics, what will inspire A. Connes introduction of non-commutative geometry in the 1980's; one could mark the official foundation of the field to his 1985 seminal paper \emph{Non-commutative Differential Geometry} \cite{Connes85} (although many ideas already appear in his 1980 note in the Comptes Rendus \cite{connes1980c}). But we need to backtrack a little in order to understand the main ideas of that paper.

\section{Matrix Mechanics}\label{MM}
In 1925, M. Born, V. Heisenberg and P. Jordan proposed a foundational framework for quantum mechanics in their classical papers \emph{Zur Quantenmechanik} \cite{born1925quantenmechanik} (both Born and Heisenberg received Nobel prizes in part for this work). The main  insight was Heisenberg's and it is the stuff of legend: in order to escape an attack of hay fever, on June 7, 1925, Heisenberg leaves G\"ottingen for Helgoland and between learning poems by Goethe and climbing, he realized that by positing \emph{non-commuting observables}, he can solve the observational puzzle of the behaviour of the spectral lines of hydrogen. After a sleepless night of calculation, he was so deeply shaken by the result that he left the house and awaited the sunrise on top of a rock. Heisenberg rapidly wrote a paper using this insight but, afraid of its own originality, he first asked Born to look at it; it didn't take long for Born to realize that he could make sense of it in the language of matrices making it more palatable.

Let's explain (following Connes \cite{Connes94}) how one would reach Heisenberg's conclusion (physical observables must be modeled by a non-commutative algebra) from the available experimental evidence at the end of the XIX century. The new phenomena being analyzed by then, and for which classical mechanics was naturally being used as a model, was the behavior of the realm of the very small, namely, atoms. The evidence came from spectrometry: by heating a tube filled with a certain chemical element, and decomposing the light it emits (perhaps by using a prism) into its various frequencies, we obtain thus a number of lines of light that is convenient to index by their wavelenghts. Such data -- the list of wavelenghts -- is what is known as an atomic spectrum; for example, for the hydrogen atom, by performing such experiment, one obtains as its spectrum the following ordered subset of $\mathbf{R}$:
$$
\left\{\frac{9}{5} L, \frac{16}{12} L, \frac{25}{21} L, \frac{36}{32} L\right\}
$$
where $L=36.456 \times 10^{-8}m$; here, we are considering wavelengths and not frequencies. In other words, we get the quantized (discrete, with an integral structure) spectrum 
$$
\left\{ \lambda =\frac{n^2}{n^2-4} L: n=2,4,5,6,\right\}
$$

For a more complicated atom, the experimental outcome is:
$$
\frac{1}{\lambda}=\frac{R}{m^2}-\frac{R}{n^2}
$$
where $R=4/L$ (Rydberg, 1890), $m$ is a fixed integer, and $n$ takes certain integer values. Again, we obtain a discrete spectrum.

Now, this experimental result contradicts classical physics; indeed, it is not very hard to show that Newtonian mechanics coupled with Maxwell's theory of electromagnetism imply together that the spectrum $\Gamma$ of such an atom should be an additive subgroup $\Gamma \subseteq \mathbf{R}$ of the real line. 

But what the experiments prescribe is the Ritz-Rydberg combination principle. If we choose frequency $\nu/\lambda$ as the more natural parameter for the spectral lines, then, there exists a discrete set $I:=\{i,j,\ldots \}$ of labels for the frequencies such that the spectrum is the set:
$$
\left\{ \nu_{ij}:=\nu_i -\nu_j|(i,j) \in I \times I \right\}.
$$
In this language, the Ritz-Rydberg combination principle states that:
$$
\nu_{ik}=\nu_{ij}+\nu_{jk},
$$
namely, the sum of \emph{certain} frequencies in the spectrum stay in the spectrum (we have a partially defined combination rule); $\nu_{ij}$ and $\nu_{lk}$ combine only when $j=l$. 

To explain this better, it is convenient to define two  structural maps; first the source map: 
$$ s:I\times I \to I, \ \ \  s(i,j):=i,$$
and  then the target map:
$$ t:I\times I \to I, \ \ \  t(i,j):=j.$$
The algebraic object:
$$\mathcal{I}:=(I,I\times I,s,t).$$

We must think of $(i,j)$ as an arrow connecting $i$ to $j$ and $I\times I$ as a set of arrows. Two arrows $(i,j)$ and $(l,k)$  can be composed (or ``multiplied"), but only whenever $j=l$:
$$(j,k)\circ (i,j):=(i,k).$$
This multiplication admits as many ``left identities" as elements $j$ of $I$ exist, for:
$$(j,j)\circ (i,j):=(i,j),$$
and the same goes for ``right identities".

$\mathcal{I}$ is what we will call a \emph{groupoid}, one of the main bulding blocks in non-commutative geometry.

\section{Groupoids}\label{SectGpd}

To make full sense of this result we will need the notion of a (small) groupoid $\mathcal{G}$ and of its convolution algebra $C(\mathcal{G})$. Groupoids are certain kind of algebraic objects: on the one hand, they generalize groups; on the other, they are  categories where every arrow has an inverse:

\begin{definition} A \emph{category} $\mathcal{C} =(C_0, C_1,s,t)$ consists of objects $C_0$ and arrows $C_1$ together with maps:
\begin{itemize}
	\item[a)] Two maps, $s:C_1\rightarrow C_0$ and $t:C_1\rightarrow C_0$, called \emph{the source map} and \emph{the target map} so that if for an arrow $\alpha\in C_1$ we have, $s(\alpha)=x$ and $t(\alpha)=y$, then we write 
	$$
	x \stackrel{\alpha}{\longrightarrow}y
	$$
	or
	$$
	\alpha : x\longrightarrow y.
	$$
	\item[b)] The \emph{identity-arrow} map
	$$
	i:C_0 \rightarrow C_1
	$$
	assigning to every object $x\in C_0$ its identity arrow
	\begin{align*}
		i(x) &: x\longrightarrow x,\\
		1_x&:= i(x).
	\end{align*}
	\item[c)] A composition law for arrows: 
	$$
	m:{C_1} \tensor[_t]{\times}{_s} C_1 \longrightarrow C_1
	$$
	Here ${C_1} \tensor[_t]{\times}{_s} C_1$ consists of pairs of arrows $(\alpha, \beta)$ so that $t(\alpha )=s(\beta)$. The composition law is only \emph{partially defined}, namely, the domain of $m$ is not all of $C_1\times C_1$ but only the subset ${C_1} \tensor[_t]{\times}{_s} C_1$.
	For the composition law, we often write:
	$$
	\beta \circ \alpha:= m(\alpha, \beta).
	$$
	These data satisfy the two strong algebraic conditions:
	\begin{itemize}
		\item[i)] Associativity:
		$$
		\alpha \circ (\beta \circ \gamma) = (\alpha \circ \beta) \circ \gamma,
		$$
		whenever $s(\alpha =t(\beta)$ and also $s(\beta) =t(\gamma)$. In other words, in the following commutative diagram, $\delta$ is well defined:
		
		\centerline{\xymatrix{
			x \ar@{->}[r]^\alpha & y \ar[d]^\beta\\
			w  \ar@{->}[u]_\delta & z \ar[l]^\gamma
		}}
		\item[ii)] Identity: if
		$$
		\alpha : x\longrightarrow y
		$$
		Then
		$$
		\alpha \circ 1_x = 1_y\circ \alpha =\alpha.
		$$ 
	\end{itemize}
\end{itemize}
	
\end{definition}

\begin{definition}
	Given two categories $(\mathcal{C},C_0,C_1,s,t)$ and $(\mathcal{C'},C'_0,C'_1,s',t')$ a functor $F:\mathcal{C}\to\mathcal{C}'$ is a rule assigning objects in $\mathcal{C}$ to objects in $\mathcal{C}'$ and arrows in $\mathcal{C}$ to arrows in $\mathcal{C}'$ and satisfying $F\circ s = s' \circ F$, $F\circ t = t' \circ F$, $F\circ i= i' \circ F$ and:
	$$ F(\alpha \circ \beta) = F(\alpha) \circ F(\beta).$$ 
\end{definition}

\begin{definition}
	Given objects $x,y\in C_0$, we say that $x$ is isomorphic to $y$ in $\mathcal{C}$ and write $x\cong y$ if the is an arrow $\alpha:x\to y$ together with an inverse arrow $\alpha^{-1}: y \to x$, that is to say, we have: $\alpha\circ\alpha^{-1} = 1_y, \ \ \ \alpha^{-1}\circ\alpha = 1_x.$ 
\end{definition}

\begin{example}
   The category $\mathcal{S}=(S_0, S_1,s,t)$ of all sets where $S_0$ is the class of all sets and $S_1$ the class of all mappings of sets. For two sets $x,y$, and a mapping:
   $$ \alpha:x\longrightarrow y,$$%
   we set $s(\alpha)=x$ and $t(\alpha )=y$. It is immediate to verify all the necessary algebraic conditions, and so, $\mathcal{S}$ is a category.

   This category has many subcategories of importance, for example, the category of all groups $\mathcal{G}=(G_0, G_1)$, where $G_0$ is the class of all groups and $G_1$ is the class of all group homomorphisms.
\end{example}

Category theory was discovered by S. Eilenberg and S. Mac Lane in the years of 1942-1945 and first appeared fully formed in their 1945 classical paper \emph{General Theory of Natural Equivalences} \cite{EilenbergMacLane} (very much under the influence of E. Noether, one of Mac Lane's teachers). In this work, the concept of category was mostly auxiliary, for the developments of homological algebra in algebraic topology had motivated Eilenberg and Mac Lane to understand systematically the concept of natural transformation.

Perhaps the best way to think of a natural transformation $\eta : F \Rightarrow F'$ from a functor $F: \mathcal{A} \to \mathcal{B}$ to a functor $F': \mathcal{A} \to \mathcal{B}$ is as a homotopy from $F$ to $F'$. To make sense of this, it is useful to define the \emph{unit interval category} $\mathcal{J}$  as the category having two objects $0$ and $1$ and three arrows (only one being a non-identity arrow) depicted below:

 \ 
 
 \centerline{\xymatrix{
 		0 \ar@(ul,ur)^{1_0} \ar@{->}[r] & 1 \ar@(ul,ur)^{1_1}
 }}

\ 

\begin{definition}
	A \emph{natural transformation} $\eta$ from $F$ to $F'$ is a functor 
	$$
	\eta : \mathcal{A} \times \mathcal{J} \longrightarrow \mathcal{B},
	$$
	so that $F= \eta|_{\mathcal{A}\times 0}$ and 
	$F'= \eta|_{\mathcal{A}\times 1}$.
	When such a natural transformation exists, we write $\eta:F \Longrightarrow F'$. 
\end{definition}

It is natural to define the composition (concatenation) $\eta\circ \eta'$ of natural transformations by considering the \emph{double interval category} which contains $\mathcal{J}$:

\

 \centerline{\xymatrix{
		0 \ar@(ul,ur)^{1_0} \ar@{->}[r] \ar@/_2pc/[rr]^\iota & 1/2 \ar@(ul,ur)^{1_{1/2}} \ar@{->}[r]& 1 \ar@(ul,ur)^{1_1}
}}

\

\begin{definition} 
	We say that $(\eta:F\to F',\eta':F'\to F)$ are a \emph{natural equivalence} (homotopy equivalence) of categories, and we write $\mathcal{A} \simeq \mathcal{B}$ if the concatenations $\eta \circ \eta'$ and $\eta'\circ\eta$ send $\iota$ to the identity transformations $1_F$ and $1_{F'}$ respectively. 
\end{definition}

Just as homotopy equivalet spaces do not need to have the same cardinality (the uncountable unit disc is homotopy equivalent to a point), equivalent categories can have vastly different number of objects. In fact, intuitively, if $\mathcal{A} \simeq \mathcal{B}$, then $\mathcal{B}$ can be obtained from $\mathcal{A}$ by means of a intermediate category $\mathcal{C}$: 

 \ 
 
 \centerline{\xymatrix{
		\mathcal{A}  &\  \mathcal{C}\  \ar@{->}[l] \ar@{->}[r]& \mathcal{B} 
}}

\ 

Both arrows induce equivalences, and the left arrow (resp. the right arrow) can be obtained from $\mathcal{A}$ (resp. $\mathcal{B}$) by deleting objects of $\mathcal{A}$ (and all arrows starting or ending in the deleted object) (resp. deleting objects of $\mathcal{B}$) making sure that $\mathcal{C}$ still has, at least, one object in every isomorphism class of objects in $\mathcal{A}$ (resp. $\mathcal{B}$); the left arrow thins out $\mathcal{A}$, and the second arrow fats up $\mathcal{C}$ to obtain $\mathcal{B}$. The diagram above is important, for it is an archetype for non-commutative geometry: we will see this later, when we talk about bi-bundles.

\begin{example}
	Consider the category $\mathcal{V}$ of complex $n$-dimensional vector spaces together with linear isomorphisms. It is not hard to see that this category is equivalent to the category $[\bullet/\GL_n(\C)]$ which has just one (abstract) object $\bullet$, and $n\times n$ invertible matrices as arrows with multiplication as its composition law. Notice that (by definition) every arrow in both categories has an inverse.
\end{example}

\begin{definition}
	A category $\GG=(G_0,G_1,s,t)$ in which for every arrow $\alpha:x\to y$ there exists an inverse arrow $\alpha^{-1}: y \to x$, namely an arrow so that: 
	$$\alpha\circ\alpha^{-1} = 1_y, \ \ \ \alpha^{-1}\circ\alpha = 1_x,$$
	is called a \emph{groupoid}.
\end{definition}

\begin{example}
	Every group $G$ can be made into a groupoid $[\bullet/G]:=(\{\bullet\},G,s,t)$ (for $s$ and $t$ the constant maps $G\to\{\bullet\}$) by considering the category $[\bullet/G]$ with one (abstract) object $\bullet$ and an arrow $\tilde{g}$ for every element $g\in G$. Given two arrows $\tilde{g}:\bullet\to\bullet$ and $\tilde{h}:\bullet\to\bullet$ (for $h,g\in G$), we define:
	$$ \tilde{g}\circ\tilde{h} := g\cdot h.$$ 
\end{example}

\begin{example}
	Every equivalence relation can be made into a groupoid. Consider a set $I$ and $R\subseteq I\times I$ an equivalence relation on $I$ ($R$ is the set of pairs $(i,j)$ so that $i$ is related to $j$). 
	Then, we can define a groupoid $[I/R]:=(I,R,s,t)$ writing $s(i,j):=i$, $t(i,j):=j$ and
	$$(i,j)\circ(j,k):=(i,k).$$ The verification of the claim that $[I/R]$ is a groupoid is immediate.
	
	\emph{The case $R=I\times I$ constructs the groupoid $[I/I\times I]$ arising from matrix mechanics at the end of last section}. In fact, the Ritz-Rydbergh combination principle can be interpreted as saying that the frequencies $\nu:[I/I\times I]\to [\bullet/\real]$ are the image of a real valued groupoid representation (and no longer a group representation). This implies that the space-momentum coordinates in the microscopic phase-space do not commute, as Heisenberg discovered.
\end{example}

\begin{example}
	Every group action $G\times M \to M$ of $G$ on $M$ can be made into a \emph{translation groupoid} $[M/G]:=(M,M\times G,s,t)$ by writing $s(m,g)=m$, $t(m,g):=g\cdot m$ and 
	$$(gm,h)\circ(m,g) := (m,hg).$$
\end{example}

For the purposes of geometry, it is useful to restrict our attention to small categories (which do not include the category of sets).

\begin{definition}
	We say that a category $\mathcal{C}=(C_0,C_1,s,t)$ is small if both $C_0$ and $C_1$ are sets.
\end{definition}

\begin{definition} 
	Given an object $x$ in $C_0$ for $\mathcal{C}$ a small category, the set of invertible arrows $g : x\to x$ forms a group called \emph{the automorphism group} of $x$ in $\mathcal{C}$.
\end{definition}

The main source of non-commutative spaces are groupoids that have a geometric structure, namely, topological and Lie groupoids.

\begin{definition}
	A topological (resp. Lie) groupoid is a small groupoid $\GG=(G_0,G_1,s,t)$ so that $G_0$ and $G_1$ are topological spaces (resp. Hausdorff smooth manifolds) and all structure maps $s$,$t$,$m$,$i$ are continuous (resp. smooth).
\end{definition}

Small discrete groupoids (where both $G_0$ and $G_1$ are discrete topological spaces) are really not more general than discrete groups asthe following example explains.

\begin{example}
	Consider a small discrete groupoid $\GG=(G_0,G_1,s,t)$, and using the axiom of choice, pick exactly one object in $G_0$ for every isomorphism class of objects in $\GG$. If we set $I:=G_0/\cong$ as the indexing set of isomorphism classes, we can write for such set of objects $(x_i)_{i\in I}$. Then it is not hard to see that $\GG$ is equivalent to the disjoint union $\amalg_i [\bullet/G_i]$, where $G_i$ is the group of automorphisms of $x_i$ in $\GG$.
\end{example}

From now on we will restrict our attention to Lie groupoids.

\begin{definition}
	We say that a smooth map of manifolds $f:M\to N$ is \emph{\'{e}tale} if it is a local diffeomorphism; that is to say $f$ is both a submersion and an immersion.
	We say that  $\GG=(G_0,G_1,s,t)$ is an \emph{\'{e}tale} Lie groupoid if $s$ is \'{e}tale.
\end{definition}

In fact, the main examples that we will consider in this note (foliation groupoids) can be made to be \'{e}tale \cite{CrainicMoerdijk, MoerdijkMrcun} (e.g. the non-commutative torus below).

\begin{example}
	A Lie groupoid $\GG:=[M/G]$ (usually called a \emph{translation groupoid}) is \'{e}tale whenever $G$ is discrete. 
\end{example}

\begin{example}
	A choice of an atlas $(U_i)_i$ for a manifold $M$, gioves rise to an \'{e}tale groupoid $\UU:=(\amalg_i U_i, \amalg_{(i,j)} U_{ij},s,t)$, where 
	\begin{itemize}
		\item $\amalg_i U_i:=\{(m,i)| m \in U_i\},$
		\item $\amalg_{(i,j)} U_{ij}:= \{(m,i,j) | m\in U_i\cap U_j\},$
		\item $s(m,i,j):=(m,i),$
		\item $t(m,i,j):=(m,j),$
		\item $(m,j,k)\circ (m,i,j):= (m,i,k).$
	\end{itemize}
\end{example}

We need a geometric version of the equivalence of groupoids that corresponds to the equivalence of categories of the previous section:

\begin{definition}
	 Given two Lie groupoids $\HH=(H_0,H_1,s,t)$ and $\GG=(G_0,G_1,s,t)$, a morphism $\phi_i : H_i \to G_i$, $i=0,1$, is an \emph{essential equivalence} if 
	\begin{itemize}
		\item[i)] $\phi$ induces a surjective submersion $(y,g) \mapsto t(g)$ from $$H_0 \times_{G_0} G_1 =\{ (y,g) | \phi(y)=s(g) \}$$ onto $H_0$; and
		\item[ii)] $\phi$ induces a diffeomorphism $h \mapsto(s(h)\phi(h),t(h))$ from $H_1$ to the pullback $H_0 \times_{G_0} G_1 \times_{G_0} H_0.$
	\end{itemize}
	We say that two Lie groupoids $\GG'$ and $\GG$ are \emph{Morita equivalent} if there exists a Lie groupoid $\HH$ and two essential equivalences $ \GG \leftarrow \HH \rightarrow \GG'$ (and we will say that $\HH$ is a $\GG$-$\GG'$-bi-bundle).
	The equivalence class $\bar{\GG}$ of the groupoid $\GG$ under Morita equivalence is called the \emph{$C^{\infty}$-stack associated to $\GG$}.
\end{definition}

\begin{example} 
	Given a fixed manifold $M$ and two atlases $(U_i)$ and $(V_j)$, then the two associated \'{e}tale groupoids $\UU$ and $\VV$ are Morita equivalent if and only if the atlases are equivalent in the atlas sense. Thus, $M$ itself is the stack associated to $\UU$ (and $\VV$): 
	$$ M \cong \overline{\UU} \cong \overline{\VV}.$$  
\end{example}

\begin{example}
	Consider a foliated manifold $(M,\FF)$ with $q$  the codimension of the foliation. The holonomy (or foliation) groupoid $\HH = \mathrm{Holo}(M,\FF)$ has as objects $H_0 = M$, and two objects $x,y$ in $M$ are connected by an arrow if and only if they belong to the same leaf $L$; arrows from $x$ to $y$ are in correspondance to homotopy classes of paths lying on $L$ starting at $x$ and ending at $y$.  The foliation groupoid $\HH = \mathrm{Holo}(M,\FF)$  is always Morita equivalent to an \'{e}tale groupoid for if we take an embedded $q$-dimensional transversal manifold $T$ to the foliation that hits each leaf at least once then the restricted groupoid $\HH|_T$ is an \'{e}tale groupoid, and, moreover, it is Morita equivalent to $\HH = \mathrm{Holo}(M,\FF)$ \cite{MoerdijkMrcun}.
\end{example}

\section{Convolution Algebras}

It is time to explain how to obtain a non-commutative algebra out of a groupoid.

\begin{definition}
	Given an \'{e}tale groupoid $\GG$, we associate to it a non-commutative algebra $A_\GG$, the \emph{convolution algebra} of $\GG$; its elements are compactly supported smooth complex valued functions on the manifold $G_1$ of arrows of $\GG$,  $f\colon G_1 \to \complex$. The \emph{convolution product} $f * g$ of two functions is given by:
	$$(f*g)(\alpha) = \sum_{\beta\circ\gamma=\alpha} f(\beta)g(\gamma),$$ where the sum is well defined because it ranges over a discrete space ($\GG$ is \'{e}tale) and finite because the functions are required to be compactly supported. The algebra $A_\GG$ can be made into a $C^*$-algebra. In general, $A_\GG$ is a non-commutative algebra.
\end{definition}

\begin{example}
	Consider a (discrete) group $G$, the convolution algebra of the groupoid $[\bullet/G]$ is exactly the same as the group algebra of $G$.
\end{example}

\begin{example}
	Consider now the Heisenberg groupoid $[I/I\times I]$ from matrix mechanics. Its convolution algebra is a matrix algebra:
	$$A_{[I/I\times I]}\cong \mathrm{Mat}_{n\times n}(\complex),$$
	where $n$ is the cardinality of $I$.
\end{example}

The category of all categories is actually a \emph{2-category}: it has objects, and for every pair of objects $x$, $y$, the family of arrows going from $x$ to $y$ is itself a category. An arrow $\eta: \alpha\to\beta$ between arrows is referred to as a $2$-arrow.

There are two 2-categories that are of great importance in non-commutative geometry: the 2-category of groupoids and the 2-category of algebras. Due to space considerations, I am all but ignoring the analytical issues concerning $C^*$-algebras, which is too bad for it is a very important ingredient in the field; in any case, we will be working only at the formal level from now on.

The 2-category of groupoids $\Groupoids$ has:
\begin{itemize}
	\item[G1)] Objects: groupoids.
	\item[G2)] Arrows: (smooth) functors.
	\item[G3)] 2-arrows: natural transformations.
\end{itemize}

The 2-category of algebras $\NCAlgebras$ has:
\begin{itemize}
	\item[A1)] Objects: associative (possibly non-commutative) algebras.
	\item[A2)] Arrows: bimodules over algebras.
	\item[A3)] 2-arrows: bimodule morphisms.
\end{itemize}

Observe that a morphism $A \to B$ of algebras is not an algebra homomorphism but rather a bi-module ${}_A M_B$. The composition of two arrows (bimodules) is given by:
$$ {}_B M_C \circ {}_A M_B:= {}_A M_B \otimes {}_B M_C.$$ 
The notion of isomorphism of algebras in this category is called Morita equivalence of algebras.

\begin{definition}
	Two algebras $A$ and $B$ are \emph{Morita equivalent} iff there is an $A$-$B$-bimodule $M$, and a $B$-$A$-bimodule $N$ so that $M\otimes_B N \cong A$ (as $A$-$A$-bimodules), and $N\otimes_A M \cong B$ (as $B$-$B$-bimodules). Equivalently, $A$ and $B$ are Morita equivalent if and only if their categories of modules $A$-$\mathrm{Mod}$ and $B$-$\mathrm{Mod}$ are equivalent.
\end{definition}

\begin{example}
	Two commutative algebras are Morita equivalent iff they are isomorphic.
\end{example}

The important point \cite{Mrcun} is that there is a convolution 2-functor:
$$\Groupoids \longrightarrow \NCAlgebras,$$
that, when restricted to objects, sends $\GG$ to its convolution algebra $A_\GG$.

This implies immediately that (for \'etale groupoids) if the groupoid $\GG$ is Morita equivalent to $\GG'$ (as groupoids), then the algebra $A_\GG$ is Morita equivalent to $A_{\GG'}$ (as algebras): the Morita equivalence class $\bar{A_\GG}$ only depends on the stack $\bar{\GG}$ and not on the groupoid. But two completely different stacks could have the same convolution algebra.

\begin{example}
	Given a compact manifold $M$ and an atlas $(U_i)$, the (non-commutative) convolution algebra $A_\UU$ of the groupoid $\UU$ associated to the atlas is Morita equivalent to $C(M)$ the algebra of smooth complex valued functions on $M$ which is commutative.
\end{example}

\begin{example}
	Consider the groupoids $\GG_1=[\bullet/\integer]$ and $\GG_2=[\integer/\{1\}]$. The first one is connected, while the second has infinitely many components; therefore, the groupoids $\GG_1$ and $\GG_2$  are not Morita equivalent; nevertheless the Fourier transform $\FF\colon A_{\GG_1} \to A_{\GG_2}$ is an isomorphism and, therefore, a Morita equivalence. This shows that the convolution 2-functor forgets information. This is a feature rather than a bug in non-commutative geometry. 
\end{example}

\begin{example}
	The Heisenberg groupoid $[I/I\times I]$ is Morita equivalent to the trivial groupoid $[\bullet/\{1\}]$; therefore, the non-commutative matrix algebra $\mathrm{Mat}_{n\times n}(\complex)$ is Morita equivalent to the 1-dimensional commutative algebra $\complex$.
\end{example}

\section{Gelfand Duality}\label{GD}

Gelfand duality expresses the fact that it is the same thing to have spaces as it is to have commutative algebras:

\begin{theorem}[Gelfand Duality] 
	The categories $\Spaces$ of Hausdorff compact topological spaces and the opposite category to the category $\CommAlgebras$ of commutative $C^*$-algebras are equivalent. Given a topological space $X$, its corresponding algebra is the algebra $C(X)$ of continuous complex valued functions on $X$ with pointwise multiplication.
\end{theorem}

\begin{remark}
Given a category $\mathcal{C}$, its opposite category $\mathcal{C}^\op$ has the same objects and the same arrows but $s$ and $t$ exchange roles so that, in $\mathcal{C}^\op$, we have that $s^\op=t$ and $t^\op=s$.
\end{remark}

In classical algebraic geometry, one starts with an affine variety $X$ and one produces a commutative algebra $\OO(X)$ by taking its regular functions. Then, one can go back to $X$ by taking the spectrum of maximal ideals of $\OO(X)$. A similar but more delicate construction works in the case of a topological space $X$.

\begin{remark} The category $\CommAlgebras$ is the same as the category $\CommAlgebras/{\sim_M}$ where we have inverted Morita equivalences as two commutative algebras are Morita equivalent iff they are isomorphic.
\end{remark}

We are finally ready to define non-commutative spaces.

\begin{definition}
	The category of non-commutative spaces $\NCSpaces$ is the opposite to the category $\NCAlgebras/{\sim_M}$ of possibly non-commutative algebras \emph{up to Morita equivalence}.
\end{definition}

This definition extends Gelfand duality into the non-commutative realm:

\begin{center}
	\begin{tikzcd}
	\Spaces\arrow[r, "\cong", leftrightarrow]\arrow[d, hook]&
	\CommAlgebras^\op \arrow[d, hook] \\
	\NCSpaces\arrow[r, "\cong", leftrightarrow]& (\NCAlgebras/{\sim_M})^\op
	\end{tikzcd}
\end{center}

Also, the convolution functor becomes a well defined functor:
$$\Stacks \longrightarrow \NCSpaces.$$

In fact, we have:

\begin{center}
	\begin{tikzcd}
	\Groupoids\arrow[r, "C", rightarrow]\arrow[d, twoheadrightarrow]&
	\NCAlgebras\arrow[d, twoheadrightarrow] \\
	\Stacks\arrow[r, "C", rightarrow]& \NCSpaces
	\end{tikzcd}
\end{center}
where $\Stacks\cong\Groupoids/{\sim_M}$ and  $\NCSpaces\cong\NCAlgebras/{\sim_M}$.

\section{Non-Commutative Topology}

Because the rational algebraic topology of a commutative space can be written in terms of its commutative algebra, this allows one to speak of non-commutative rational topology: all the concepts that generalize will depend only on the Morita equivalence class of a (possibly non-commutative) algebra. We will follow \cite{kontsevich2008xi} in this section.

So, we consider  $A$ to be a unital, associative, possibly non-commutative algebra. 

\begin{definition}
	The Hochschild
	complex $C_\bullet (A,A)$ of $A$ is a negatively graded complex (we wil have all differentials of degree $+1$):
	$$ \stackrel{\partial}{\longrightarrow} A\otimes A\otimes A \otimes A \stackrel{\partial}{\longrightarrow}
	A \otimes A \otimes A \stackrel{\partial}{\longrightarrow} A\otimes A \stackrel{\partial}{\longrightarrow} A,$$
	where $A^{\otimes k}$ lives on degree $-k+1$.
	The differential $\partial$ is given by
	$$\partial(a_0 \otimes \cdots \otimes a_n) = a_0 a_1 \otimes a_2 \otimes \cdots \otimes a_n - a_0 \otimes a_1 a_2 \otimes \cdots \otimes a_n $$
	$$+ \ldots + (-1)^{n-1} a_0 \otimes a_1 \otimes \cdots \otimes a_{n-1} a_n + (-1)^n a_n a_0 \otimes a_1 \otimes \cdots \otimes a_{n-1}.$$
	The terms of this formula are meant to be written cyclically:
	\begin{equation}
	\begin{array}{ccccccc} &  &  & a_0 &  &  &  \\
	&  & \otimes &  & \otimes &  &  \\  & a_n &  &  &  & a_1 &  \\  & \otimes
	&  &  &  &  \otimes & \\  & \vdots &  &  & &   \vdots  & \\ &  & \otimes & & \otimes &  & \\  & & & a_i & & &\end{array}
	\end{equation}
	for $a_0 \otimes \cdots \otimes a_n$. It is immediate to check that $\partial^2=0$. We write 
	$$HH(A,A):=\mathrm{Ker\ }\partial/\mathrm{Im\ }\partial.$$
\end{definition}

We can interpret the homology of the Hochschild complex in terms of homological algebra:
$$ HH(A,A) = \mathrm{Tor}_\bullet^{A\otimes_k A^{\mathrm{op}} - \mathrm{mod}}(A,A).$$

It is an idea of A. Connes that, in non-commutative geometry, the Hochschild homology of $A$ can be interpreted as the complex of differential forms:

\begin{theorem}[Hochschild-Konstant-Rosenberg, 1961, \cite{HKR}] Let $X$ be a smooth affine algebraic  variety, then if $A=\OO(X)$,
	we have: $$HH_i(X):=H^{-i}(C_\bullet(A,A);\partial) \cong \Omega^i(X)$$ where $\Omega^i(X)$ is the space of $i$-forms on $X$.
\end{theorem}

\begin{proof}
Write the diagonal embedding $X\stackrel{\Delta}{\longrightarrow} X \times X$ and, because the normal bundle of $\Delta$ is the tangent bundle of $X$, we have:
$$HH_\bullet(X) = \mathrm{Tor}_\bullet^{{\mathrm{Quasi-coherent}}(X\times X)} (\OO_\Delta,\OO_\Delta).$$ A local calculation finishes the proof.
\end{proof}

The Hochschild-Konstant-Rosenberg theorem allows us to interpret $HH_i(A)$ as the space of differential forms of degree $i$ on a non-commutative space.

Whenever $A$ is non-commutative, we have:
$$H^0(C_\bullet(A,A);\partial)=A/[A,A].$$ 

In the commutative case
$A=\OO(X)$,  to an element $a_0 \otimes\cdots \otimes a_n$ in
$C_\bullet (A,A),$ the corresponding differential form is: $\frac{1}{n!}
a_0 da_1\wedge \ldots \wedge da_n.$

It is convenient to mention a reduced version of the complex $C^\redu_\bullet(A,A)$ that computes the same cohomology; it is obtained by reducing modulo
constants all terms but the first: $$\longrightarrow A \otimes A/({\bf k}\cdot 1) \otimes A/({\bf k}\cdot 1)
\longrightarrow A \otimes A/({\bf k}\cdot 1) \longrightarrow A.$$

Alain Connes' observed that we can write a formula for an additional differential
$B$ on $C_\bullet(A,A)$ of degree $-1$, inducing a differential on
$HH_\bullet(A)$ that is meant to be the de Rham differential:
$$B(a_0\otimes a_1 \otimes \cdots \otimes a_n) = \sum_\sigma (-1)^\sigma 1 \otimes a_{\sigma(0)}
\otimes \cdots \otimes a_{\sigma(n)}$$
where $\sigma\in \integer/(n+1)\integer$ runs over all cyclic permutations. It not hard to see
that: $$B^2=0,\ \ \ B\partial +\partial B = 0, \ \ \  \partial^2=0,$$ this we write as:
$$\xymatrix{
	\cdots \ar@/_/[rr]_\partial && A \otimes A/1 \otimes A/1 \ar@/_/[ll]_B \ar@/_/[rr]_\partial &&
	A \otimes A/1 \ar@/_/[ll]_B \ar@/_/[rr]_\partial && A \ar@/_/[ll]_B
}$$
and by computing the cohomology, this gives us a complex $(\Ker \partial / \Img
\partial ; B)$. A naive definition on the de Rham cohomology in
this context is the homology of this complex $\Ker B / \Img B.$

We can improve this by considering the negative cyclic complex $C^-_\bullet(A)$, which is a
projective limit (here $u$ is just a formal variable of degree $\deg(u)=+2$):
$$C_\bullet ^-:= (C_\bullet^\redu(A,A)[[u]] ; \partial + u B) =
\lim_{\stackrel{\longleftarrow}{N}}(C_\bullet^\redu(A,A)[u]/u^N;\partial + u B).$$

\begin{definition} The \emph{periodic complex} is defined as the inductive limit:
$$C_\bullet ^\per:= (C_\bullet^\redu(A,A)((u)) ;
\partial + u B) = \lim_{\stackrel{\longrightarrow}{i}}(u^{-i} C_\bullet^\redu(A,A)[[u]];\partial + u B).$$
\end{definition}

It is a ${\bf k}((u))$-module, and this implies that
 multiplication by $u$ induces a kind of Bott periodicity. The
resulting cohomology groups called (even, odd) periodic cyclic
homology and are written (respectively): $$HP_\even(A), \ \ \ \ \
HP_\odd(A).$$ This is the desired replacement for de Rham
cohomology.

For example, when $A=C^\infty(X)$ is considered with its nuclear Fr\'echet algebra strcuture,
and taking $\otimes$ to be the topological tensor product, then we obtain the canonical isomorphisms:
$$ HP_\even(A) \cong H^0(X,\complex) \oplus H^2(X,\complex) \oplus \cdots $$
$$ HP_\odd(A) \cong H^1(X,\complex) \oplus H^3(X,\complex) \oplus \cdots $$
\begin{theorem}[Connes, \cite{Connes85}, cf. Feigin-Tsygan, \cite{Tsygan}] If $X$ is a possibly singular affine algebraic variety
	and $X_\topo$ is its underlying topological space then:
	$$HP_\even(A) \cong H^\even(X_\topo,\complex)$$ and $$HP_\odd(A) \cong H^\odd(X_\topo,\complex)$$
	and these homologies
	are finite-dimensional.
\end{theorem}

As expected, whenever $A$ is Morita equivalent to $B$, then $HP_\bullet(A) \cong HP_\bullet(B)$; in other words, $HP_\bullet(A)$ only depends on the non-commutative space represented by $A$. 

\section{The Non-Commutative Torus}

The most basic and classical example of a non-commutative space is the non-commutative torus. It can be obtained as the convolution algebra of an \'{e}tale groupoid.

The quantum 2-torus $\TTT_{\hbar}^2\in\NCSpaces\cong\NCAlgebras/{\sim_M}$ corresponds under Gelfand duality to the algebra $A_\hbar$ generated by two (periodic) generators $X$, $Y$ that don't  commute but rather satisfy the relation:
$$XY = e^{2\pi i\hbar} YX.$$
This relation is precisely the Weyl exponential form of the commutation relation appearing in the work of Heisenberg, Born and Jordan that we mentioned before \cite{born1925quantenmechanik}.

The algebra $A_\hbar$ is only truly non-commutative when $\hbar$ is irrational; When $\hbar$ is rational, while $A_\hbar$ is non-commutative on the nose (except for $\hbar=0$), it is, in reality, Morita equivalent to the commutative algebra of an ordinary torus ($XY=YX$).

\begin{theorem}[Alain Connes \cite{connes1980c}, cf. Marc Rieffel, \cite{Rieffel}] 
$A_\hbar$ is Morita equivalent to $A_{\hbar'}$
if and only if: $$\hbar' = \frac{a \hbar +b}{c \hbar + d},\ \ \ \
\left(\begin{array}{cc}a & b \\c & d\end{array}\right)\in \mathrm{SL}_2(\integer).$$
\end{theorem}

One can also prove that: 
$$ HP_\even(A_\hbar)=H^0(T^2) \oplus H^2(T^2), $$ 
and
$$HP_\odd(A_\hbar)= H^1(T^2,\complex).$$

It is a beautiful discovery of Connes \cite{connes2008walk} that the non-commutative torus can be thought as the non-commutative space that models the space of leaves of the Kronecker foliation. The universal covering of the classical torus is the Euclidean plane, by taking the foliation of all lines of slope $\hbar$ on the plane and projecting it into the torus by the covering map, we obtain the Kronecker foliation of slope $\hbar$ on $T^2$ (Figure \ref{Kronecker}). By taking a vertical transversal circle to the foliation, it is easy to see that the holonomy groupoid of this foliation is $[S^1 /\langle \rho_\hbar \rangle]$ where $\rho_\hbar$ acts on $S^1$ by a rotation of angle $\hbar$ (cf. the think line in Figure \ref{Kronecker}). In section 6 of \cite{connes2008walk}, it is shown that the convolution algebra of $[S^1 /\langle \rho_\hbar \rangle]$ is $A_\hbar$ (it is a nice exercise using Fourier series that the interested reader may try).

\begin{figure}[h]
	\centering
	\includegraphics[width=0.7\linewidth]{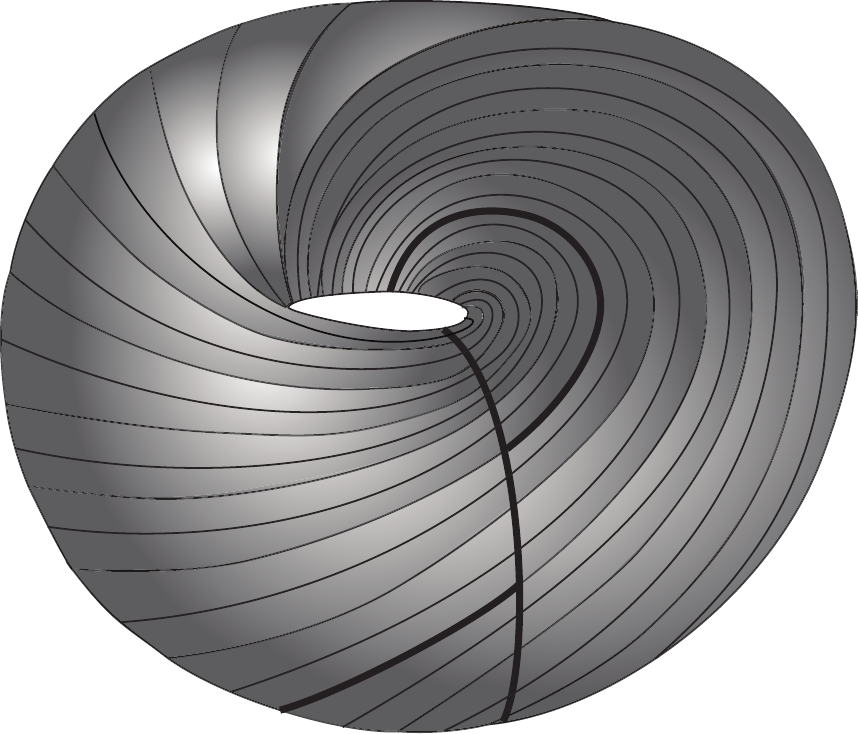}
	\caption{The non-commutative torus is the convolution algebra of the holonomy groupoid of the Kronecker foliation}\label{Kronecker}
\end{figure}

\section{Non-Commutative Toric Geometry}

Classical $n$-complex dimensional compact, projective K\"ahler toric manifolds $X$ are defined as equivariant, projective compactifications of the $n$-complex dimensional torus $\T_\complex^n:= \complex^* \times \cdots \times \C^*$:
$$X:=\overline{\T_\complex^n}.$$
An interesting question, even from a classical point of view, would be: How to define a meaningful moduli space of toric manifolds? The main problem being that toric manifolds are rigid as equivariant objects. Non-commutative geometry helps elucidate this question in a surprising beautiful way.

Let us recall first the basic facts about the \emph{moment map} defined ona toric manifold. The K\"ahler manifold $(X,g,J,\omega)$ is a symplectic manifold (forget $g$ and $J$), and the action of the real torus $\T_\R^d$ on $X$ is Hamiltonian. Therefore, we have an equivariant moment map $\mu$ with convex image $P$.
$$\mu \colon X \longrightarrow  P \subset \R^d \cong \mathrm{Lie Algebra}({\T_\R^d})^*.$$

For a toric variety $X$, $P$ happens to be a convex, rational, Delzant polytope: in other words, the combinatorial dual of $P$ is a triangulation  of the sphere $S^{d-1}$, and all the slopes of all the edges of $P$ are rational. By taking cones over the origin of the dual to the polytope, we get the \emph{fan} associated to the toric manifold \cite{Audin} (see Figure \ref{Moment}). Both $P$ and the fan live in $\rational^n$.

\begin{figure}[h]
	\centering
	\includegraphics[width=0.7\linewidth]{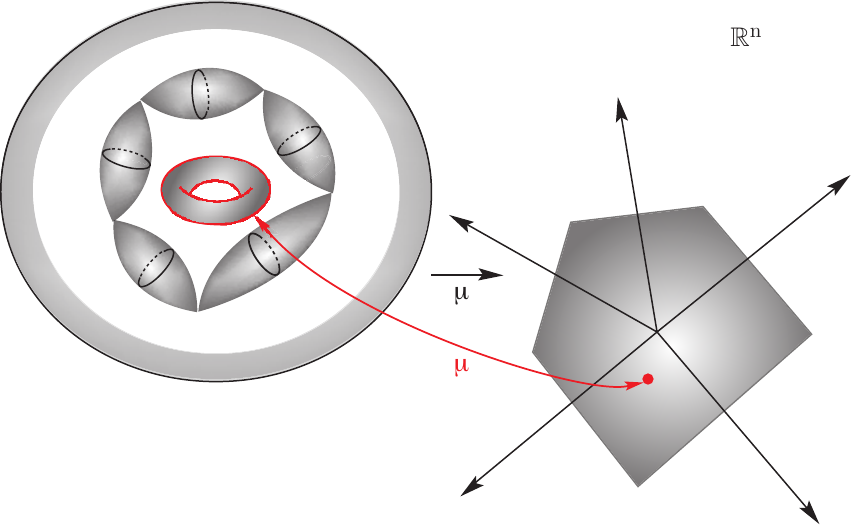}
	\caption{The moment map for a toric manifold: the inverse image of every point is a real torus of dimension equal to the dimension of the stratum of $P$ where the point lands. The inverse image of edges are spheres made up of 1-tori (circles). In non.commutative toric geometry all tori and circles are replaced by their non-commutative counterparts}\label{Moment}
\end{figure}

In \cite{katzarkov2020quantum}, classical toric geometry is generalized: \emph{by replacing all the classical tori in toric geometry for non-commutative tori, one can obtain non-commutative toric varieties}. Now the (possibly irrational) fan (or possibly irrational polytope) no longer lives in $\rational^n$, but rather  lives in (a possibly irrational) quantum lattice $\Gamma \subset \real^n$ ($\Gamma$ is a finitely generated possibly dense Abelian subgroup of $\real^n$ as it may have more than $n$ generators over $\integer$).

Then, a moduli space of toric varieties $\mathcal{M}$ can be defined (fixing the combinatorics of the polytope or fan). In  a large family of favorable cases the moduli space $\mathcal{M}$ is a complex orbifold: its rational points are precisely the classical toric varieties, and \emph{its irrational points are precisely the truly non-commutative toric varieties}. Non-commutative geometry is precisely what is needed to define a nice moduli space of toric varieties.

Just as classical toric geometry has been used in the solution of multiple problems in geometry, physics and combinatorics, non-commutative toric geometry allows many of these solutions to generalize to wider settings.

\section{Further Directions}

In this note we have hardly made justice to the richness of non-commutative geometry nor to some of its most exciting recent developments. We will mention briefly some of this beautiful mathematics; We do this by mostly following two recent excellent survey papers by Connes: \cite{ConnesEssayRH} and \cite{connes2019noncommutative}.

\subsection{Non-commutative manifolds} A classical (spin) Riemannian manifold can be seen in at least four different lights (by applying various forgetful functors): as a measure space, as a topological space, as a smooth manifold, and as a fully fledged Riemannian manifold. So far we have remained mostly at the first three levels but we have not mentioned the metrical aspect at all. In this subsection we mostly mark some pointers for the reader to fill this gap.

Connes isolated the correct definition of ``non-commutative Riemannian manifold" (and of ``non-commutative manifold") in his notion of \emph{spectral triple}. A \emph{spectral triple} $(\AAA,\HH,\DD)$ consists of an algebra $\AAA$ of operators in the Hilbert space $\HH$ (just as in Gelfand duality, when commutative, $\AAA$ is meant to be the algebra of functions of a geometric space as in section \ref{GD}), together with a (not necessarily bounded) self-adjoint operator $\DD$, acting on $\HH$ playing the role of the  inverse line element coming form the metric. 

The classical case of a spin, compact, Riemmanian manifold, $M$, can be realized as a spectral triple $(\AAA,\HH,\DD)$ by setting $\AAA$ to be the algebra of functions on $M$ acting in the Hilbert space $\HH$ of $L^2$-spinors and by letting $\DD$ be the Dirac operator. In this case, the metric on $M$ can be recovered by the following formula:
$$d(a,b)= \sup_{\{f: ||[f,\DD]||\leq 1\} } |f(a)-f(b)|.$$
Here the word ``spectral" plays two roles: the first is (at the topological level) given by Gelfand's duality -- from $\AAA$ we can recover a locally compact space --, the second role is at the geometric level, for we reconstruct the geometry from the spectrum of the Dirac operator $\DD$ (as in the motto ``can you hear the shape of a drum?") together with the interaction between the algebra $\AAA$ and the set of functions $\{f(\DD)\}$ of $\DD$. It is a remarkable fact that this works for fractal, discrete, and arc-disconnected spaces.

In the non-commutative case, replacing in the above formula $f(a)-f(b)$ by $\phi(f)-\psi(f)$ computes the distance between two states (positive linear forms) $\phi,\psi\colon \AAA\to\complex$ thus inducing a metric on the space of such states. The notion of spectral triple is an ambitious generalization on the notion of (spin) Riemmanian manifold and we refer the reader to section 2 of \cite{connes2019noncommutative} an references therein for a quick updated tour of the concept and to chapter 6 of \cite{Connes94} for a more detailed classical account. A highlight of this theory is the deep ``reconstruction theorem" of Connes (page 19 of \cite{connes2019noncommutative} and references therein) which roughly characterizes classical smooth manifolds as commutative spectral triples satisfying some additional axioms. 

Let us end this section pointing out that the spectral approach to geometry afforded by non-commutative geometry can encode simultaneously macroscopic and microscopic phenomena (the line element includes the information of all the bounding forces known so far, as in the standard model).

\subsection{The standard model of particle physics} Just as non-commutative geometry can be motivated by the most basic ideas in quantum mechanics as we did in Section \ref{MM} above, along the same lines of reasoning, non-commutative geometry can illuminate, for example, why the gauge group for the Standard Model of particles and forces is $SU(3)\times SU(2) \times U(1)$. The non-commutative geometry approach to the Standard Model was developed by Connes and his collaborators in the 90s (starting with his paper with J. Lott in \emph{Nuclear Physics B}). The basic idea is that, at the classical level, replacing the geometry of the continuum, $M$, by the (slightly) non-commutative space $M\times F$, where $F$ is a finite geometry. This addition transforms the Lagrangian of quantum electrodynamics into the Lagranginan of the Standard Model. For a lucid explanation of this bottom-up approach to the non-commutative geometrical exegesis of the Standard Model, we refer the reader to Chapter 6, Section 5 of \cite{Connes94}. It is worth mentioning here that, more recently, Chaseddine, Connes and Mokhanov have developed a more top-bottom approach to this theory with their simultanous quantization of both the fundamental class in $K$-homology and in $K$-theory. We refer the reader to \cite{connes2019noncommutative} page 4 and references therein for details on this point of view on the non-standard model.

\subsection{The Riemann Hypothesis} Connes together with his collaborators (specially Consani and Marcolli) has developed a very ambitious program to understand the meaning of the Riemann Hypothesis (RH from now on) that while deeply connected to non-commutative geometry, really is more a ``unified theory of all of mathematics" of sorts: it is more than an attack on the RH, as the RH becomes a theme in a monumental mathematical symphony. This development is undoubtedly one of the most exciting developments of 21st century mathematics. In this brief section, we barely make justice to this field by following closely the recent survey \cite{ConnesEssayRH} and describing it in a colloquial impressionistic manner (space would not allow otherwise).  We urge the interested reader to go directly to \cite{ConnesEssayRH}.

Let us recall the reader that the Riemann Hypothesis has a deceptively simple statement. First we need to define the (complex valued) Euler-Riemann zeta function $\zeta(s)$ (of one complex variable) to be the analytic continuation of 
$$ \zeta(s) = \sum_{n=1}^\infty \frac{1}{n^s},$$
(which converges only for $|s|>1$). Euler proved (\emph{Variae observationes circa series infinitas}, 1737) that for $|s|>1$:
$$ \zeta(s) = \sum_{n=1}^\infty \frac{1}{n^s} = \prod_{p} \frac{1}{1-p^{-s}},$$
where $p$ runs through all positive integer primes, thus establishing a deep connection between $\zeta(s)$ and number theory. 

We are interested in the zeroes of $\zeta(s)$. Clearly all negative even integers satisfy $\zeta(-2m)=0$. Such zeroes are called trivial zeroes of $\zeta(s)$

The RH states that \emph{every non-trivial zero $\rho$ of $\zeta(s)$ has real part equal to $\frac{1}{2}$.} As of writing this, the RH remains one of the deepest most beautiful open problems in mathematics. 
 
One possible way to organize the story of the relation of the RH and non-commutative geometry  is around the so-called explicit formul\ae, the first of which is due to Riemann:
$$\sum_n \frac{1}{n} \pi(x^{\frac{1}{n}}) = \sum_{\rho} \Li(x^\rho) + \int_{s}^{\infty} \frac{1}{t^2-1} \frac{dt}{t\log t} - \log 2.$$
Here, the left-hand side involves the primes (for $\pi(x)$ is the function that counts all primes less than $x$), while the right-hand side involves the non-trivial zeroes $\rho$ of $\zeta(s)$. The formula uses the integral logarithm $\Li(x):=\int_{0}^{x} \frac{dt}{\log t}$. 

This amazing formula establishes a sort of duality between the primes and the zeroes of $\zeta(s)$, and it begs for a geometric interpretation: Connes program could be construed as an attempt to build up the geometric setting for such an interpretation: rather than trying to  frontally assault the RH, Connes decides to take a detour in order to find the right language where to state it; and this means, in this context, finding the right geometric objects where trace formulas in the spectral geometry of the adequate space produce the explicit formula. 

This program seems to take seriously the lesson of Deligne's proof of RH in characteristic $p$, Stepanov and Bombieri found a much more elementary proof of the same theorem in 1974, but lots of beautiful mathematics would never have arisen if this more elmentary proof would have been found in the 1930s (as James Milne likes to point out). A large part of the beauty of Deligne's proof consists in the generalizations of geometry that it required (e.g: scheme, topos), many of them due to Grothendieck.   

Thus, the program could be though of as consisting of four stepping stones: algebraic geometry, trace formulas in spectral geometry, Riemann-Roch formul\ae, and absolute algebra (homotopy theory).

The first paradigm shift in the approches to the RH came in the work of Weil who, in 1940, proved the RH for curves over finite fields: the usual RH corresponds to the field $\rational$ which is not finite. Given a curve $C$ (you can think of a Riemann surface but over a finite field), there is an analogue $\zeta_C(s)=Z(C,q^{-s})$ of the classical $\zeta(s)$ and the explicit formula in this case can be written in terms of the geometry of the space $$Y_\field := \adeles_\field/\field.$$
Here $\adeles_\field$ is the ring of adeles of $\field$ (defined as the restricted product of the local fields $\field_\nu$, in turn defined as the completions of $\field$ at all its different places), and $\field$ is the global field of functions on $C$. Weil's explicit formula in terms of the geometry of $Y_\field$ makes sense for every global field: Riemann's classical explicit formula becomes a special case of Weil's explicit formula when interpreted in the global field $\rational$. Weyl's theorem for a curve $C$ over a finite field $F_q$ becomes then a consequence of the Riemann-Roch theorem over $\bar{C}\times\bar{C}$ (here $\bar{C}$ is the curve obtained from $C$ by extension of scalars from $F_q$ to its algebraic closure $\bar{F}_q$) for his explicit formula for a global field in this case can be given a cohomological interpretation. Thus, the proof requires a good notion of intersection theory and cohomology in algebraic geometry as, in the end, it amounts to proving the negativity of the self-intersection pairing for divisors of degree zero. For a very clear and brief description of this argument, we refer the reader to Section 2.3 of \cite{ConnesEssayRH}. 

The second paradigm shift (from this geometric point of view) comes when Deligne generalized this result to all smooth projective varieties over a finite field, quite a remarkable result. In his proof, he used both, ideas by Landau on the very classical subject of Dirichlet series, as well as very modern ideas at the time from Grothendieck on schemes, and \'{e}tale cohomology. Again, the notion of space must be expanded (it is insufficient to consider the foundations of algebraic geometry as developed by Weil and one must use Grothedieck's point of view) to prove RH for a larger class of global fields. Connes project for the classical RH is more in line with Weil's method than with Deligne's: it requires even more general notions of what a space is. 

Connes suggests that a full understanding of the RH will require both the notion of non-commutative space and the concept of topos (due to Grothendieck) as the ever more general definitions of what a geometric space must mean. This would entail a third paradigm shift.

By definition, a topos is a specific type of category. The archetypal example of a topos is a Grothedieck topos which are, roughly speaking, categories of sheaves over geometric spaces (or sites): one is to model the general definition of a topos on the properties of such categories. But topoi enjoy of dual nature: geometric and logic. MacLane and Moerdijk start their introduction to topos theory by saying:
\blockquote{A startling aspect of topos theory is that it unifies two seemingly wholly distinct mathematical subjects: on the one hand, topology and algebraic geometry and, on the other hand, logic and set theory. Indeed, a topos can be considered both as a “generalized space” and as a “generalized universe of sets”.}
Vickers has proposed a useful dictionary to understand the dual nature of (localic) topoi:
\begin{center}
\begin{tabular}{ |c|c| } 
	\hline
	\textbf{Geometric Space} & \textbf{Logical Theory} \\ 
	Point & Model of the Theory  \\ 
	Open Set &  Propositional Formula \\ 
	Sheaf & Predicate Formula \\
	Continuous Map & Definable Transformation \\
	\hline
\end{tabular}
\end{center}

\ 

In topos theory is very easy to incorporate group (and groupoid) actions (by constructibility). While in ordinary geometry,

 a quotient may be ill behaved, topoi as model for spaces behave very well under quotients, just as non-commutative spaces do. Indeed, given a Lie groupoid, as in Section \ref{SectGpd} above, one can both associate a non-commutative space (the Morita class of the convolution algebra of the groupoid) and a topos: the equivariant sheaves over the space of objects. When a geometric space has an atlas given by a groupoid (it is stack-like), one can associate to it both a non-commutative space and a topos (cf. page 402 of \cite{Cartier2001}, where this is explained lucidly).

One additional ingredient on what Connes calls \emph{the Riemann-Roch strategy for the RH} (which would provide a method for applying Weil's method to the original RH) that must be mentioned here is \emph{tropical geometry}. Roughly speaking, tropical geometry is a kind of commutative algebraic geometry that lives not in a field but rather in a semi-field $(\trop,\oplus,\otimes)$. Here, as a set, $\trop:=\real_+=\{x:x\geq 0\}$, and the expressions:
$$x\oplus y := \max(x,y), \ \ \ \ x\otimes y := x+y,$$
define the operations of the semi-field. $\trop$ is a semi-field rather than a field because there is no additive inverses in general. Nevertheless, one can do geometry over this semi-field and in practice this becomes a combinatorial shadow of ordinary complex geometry (taking the place of the so-called \emph{geometry over the field of one element} in many practical situations).

Connes and Consani (section 4 of \cite{ConnesEssayRH}), motivated by Soul\'{e}'s introduction of the zeta function of a variety over the field in one element,  have proved that, to understand the classical from this Riemannn-Roch point of view of the RH, one should replace the Weil space $Y_\field = \adeles_\field/\field$ defined above for a more elaborate \emph{non-commutative space}:
$$ X_\rational := \rational^\times \backslash \adeles_\rational / \hat{\integer}^\times.$$
In other words, $X_\rational$ is the quoient of the adele class space $\rational^\times \backslash \adeles_\rational$ bu the maximal compact subgroup $\hat{\integer}^\times$ of the idele class group.

This non-commutative space has a richer avatar in the form of a topos. Given a (multiplicatively cancellative) semi-ring $R$ (e.g. $\trop$), the integer positive-indexed maps $\mathrm{Fr}_n(x) := x^n$ are always injective endomorphisms (taking the place of the Frobenius from the classical Weil approach), and then the semi-group $\naturals^\times$ acts on $R$ by this Frobenius sequence of maps. We can define a ``tropical" topos associated to this action (cf. Definition 2, Section 4 of \cite{ConnesEssayRH}):

\begin{definition}
The \emph{arithmetic site} $\AAA:=(\hat{\naturals^\times},\integer_{\max} )$ is the topos $\hat{\naturals^\times}$ of sets endowed with an action of $\naturals^\times$, together with the structure sheaf of topos-semi-rings $\OO:=\integer_{\max}$ given by the Frobenius action of $\naturals^\times$.
\end{definition}

The following remarkable theorem of Connes and Consani (Section 4.3 of \cite{ConnesEssayRH}) puts everything together:

\begin{theorem}
The points of the arithmetic site $\AAA$ over $\trop$ can be canonically identified with the non-commutative space $ X_\rational = \rational^\times \backslash \adeles_\rational / \hat{\integer}^\times.$ Moreover, the action of the idele class group on $ X_\rational = \rational^\times \backslash \adeles_\rational / \hat{\integer}^\times$ corresponds at the topos level with the action on the $\trop$-indexed Frobenius automorphisms $\mathrm{Fr}_\lambda$ of $\trop$.
\end{theorem}

In analogy to Weil's approach, an adequate (cohomological) Riemann-Roch formula in this context would provide a key inequality giving a road to the classical RH. We refer the reader to \cite{ConnesEssayRH} for fuller details.

\subsection{Final remarks} Non-commutative geometry remains indomitable: its applications go from the standard model of particle physics to topological data analysis to number theory to non-commutative motives \cite{nlab:noncommutative_motive}. 

This note didn't mention non-commutative measure theory, nor the canonical time evolution of a non-commutative algebra (cf. \cite{Connes94} page 44)  interpreted by G. Segal in classical geometric language motivated by quantum field theory. It also omitted including any results of the Rosenberg-Gabriel, Kontsevich-Soibelman, Laudal, Le Bruyn and Artin-Van de Bergh's approaches to the subject and I apologize to the reader for that. 
  
Let us just finish mentioning that the ideas of Kontsevich and his collaborators motivated by homological mirror symmetry have expanded the field enormously; for example, the paper of Katzarkov, Kontsevich and Pantev \cite{MR2483750} gives a beautiful approach to non-commutative geometry via categories.

\subsubsection{Acknowledgments} I would like to thank a superb reading and excellent suggestions by the two referees. I would also acknowledge Enrique Becerra, Ludmil Katzarkov, Laurent Meersseman, Tony Pantev and Alberto Verjovsky for useful comments. I would like to thank FORDECYT (CONACYT), IMATE-UNAM, NRU HSE, RF government grant, ag. 14.641.31.000, the Institute for Mathematical Sciences of the Americas, the Simons Foundation (Homological Mirror Symmetry), the Moshinsky Foundation, the University of Geneva, the QUANTUM project from the University of Angers and the Laboratory of Mirror Symmetry HSE Moscow. Finally, my  deepest and most sincere appreciation goes to the staff at the Notices for the figure design and to Erica Flapan for a great editorial job. 

\bibliography{NCGBib}{}
\bibliographystyle{amsplain}

\end{document}